\documentclass[runningheads]{llncs}
\usepackage[T1]{fontenc}
\usepackage{cite}
\usepackage{graphicx}
\usepackage{booktabs}
\usepackage{hyperref}
\hypersetup{
  colorlinks = true,
  citecolor = blue,
  linkcolor = blue,
  urlcolor = blue
}
\usepackage{orcidlink}
\usepackage{color}

\usepackage{amsmath,amssymb}
\usepackage{cleveref}

\usepackage{tikz,tkz-euclide}
\usetikzlibrary{arrows,calc,patterns}
\usepackage{pgfplots}
\newcommand{\lab}[1]{}

\usepackage{subcaption}
\usepackage{wrapfig}

\begin{document}
\title{Toward Satisfiability Modulo Realizability}
\author{Andrew Krapivin \orcidlink{0009-0003-0227-7660} \and
Benjamin Przybocki \orcidlink{0009-0007-5489-1733} \and
Marijn J. H. Heule \orcidlink{0000-0002-5587-8801}}
\authorrunning{A. Krapivin et al.}
\institute{Carnegie Mellon University, Pittsburgh, USA\\
\email{\{akrapivi,bprzyboc,mheule\}@andrew.cmu.edu}}
\maketitle
\begin{abstract}
Problems complete for the existential theory of the reals ($\exists \mathbb{R}$) arise throughout discrete geometry. We introduce \emph{satisfiability modulo realizability}, a SAT-based approach for solving satisfiable instances of $\exists \mathbb{R}$ whose solutions correspond to realizable geometric configurations. Our method encodes an underapproximation of a geometric problem as a SAT instance over abstract order types. Since almost all abstract order types are unrealizable, naive search is infeasible. We guide the search toward realizable order types using diversity-driven sampling, partial realizability feedback, and a novel flippability heuristic that passes only limited information between components. We apply our method to discrete geometry problems and resolve an open problem by showing that the largest set of points avoiding empty convex hexagons and convex heptagons is of size 23.

\end{abstract}
\section{Introduction}

The success of SAT solving has inspired the motto that $\mathsf{NP}$ is the new $\mathsf{P}$, but the same cannot yet be said for classes beyond $\mathsf{NP}$. In this paper, we focus on problems that are complete for the class $\exists \mathbb{R}$ (pronounced ``existential theory of the reals''), which is a continuous analogue of $\mathsf{NP}$ that often arises in areas including discrete geometry, game theory, and computer algebra \cite{existential-reals}.

We develop a SAT-based program, which we call \textsf{PointSAT}, for finding solutions to problems in discrete geometry, where solutions consist of sets of points in $\mathbb{R}^2$ satisfying combinatorial constraints.\footnote{\textsf{PointSAT} is available at \url{https://github.com/andrewkrapivin/PointSAT}.} We evaluate \textsf{PointSAT} by applying it to variants of the happy ending problem posed by Esther Klein and studied by Erd\H{o}s and Szekeres~\cite{es}, a problem so named because it led to Klein and Szekeres getting married. Given a set of points in $\mathbb{R}^2$ in \emph{general position} (i.e., no three points are collinear), we say that a \emph{$k$-gon} is a subset of $k$ points in convex position, and a \emph{$k$-hole} is a $k$-gon with no points in its interior. Using \textsf{PointSAT}, we obtained the following answer to a question of Heule and Scheucher~\cite{h6-upper-bound}, who proved the upper bound, with \Cref{fig:23} showing a witnessing point set:
\begin{theorem} \label{thm:23}
    The largest point-set in $\mathbb{R}^2$ with no 6-hole or 7-gon is of size 23.
\end{theorem}
In addition to proving \Cref{thm:23}, we demonstrate the versatility of \textsf{PointSAT} by using it to find constructions for two related problems. 

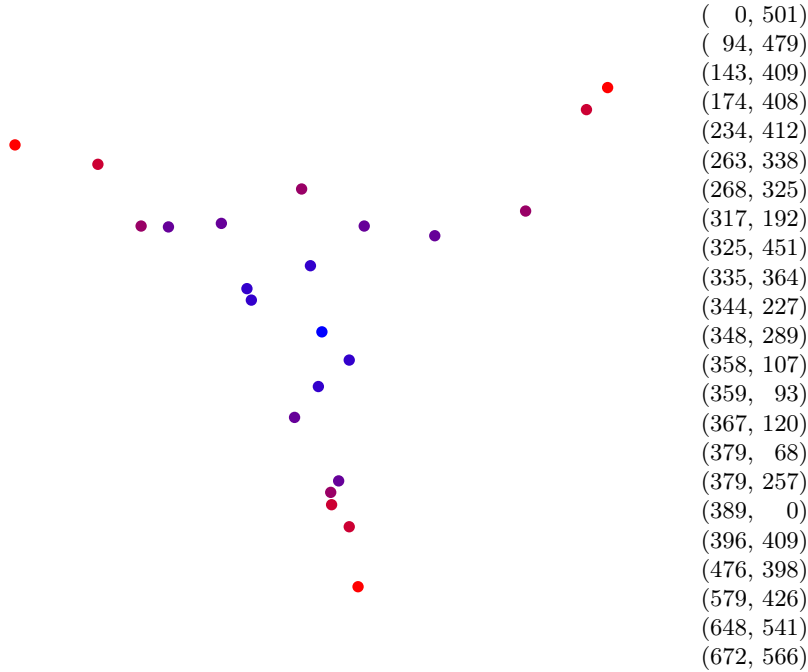
\begin{figure}[t]
    \centering
    \begin{tikzpicture}[dot/.style={inner sep=1pt, circle, minimum size=1.5mm, thick, scale=1.0}]
    \begin{axis}[
      unit vector ratio=1 1 1,
      xmin=-10,xmax=682,
      ymin=-10,ymax=576,
      width=0.8\textwidth,
      axis lines=none,
      xticklabel=\empty,
      yticklabel=\empty ]

\node[dot, fill=blue!0!red] (1) at (axis cs:389, 0) {\lab{1}};
\node[dot, fill=blue!0!red] (2) at (axis cs:672, 566) {\lab{2}};
\node[dot, fill=blue!20!red] (3) at (axis cs:648, 541) {\lab{3}};
\node[dot, fill=blue!40!red] (4) at (axis cs:579, 426) {\lab{4}};
\node[dot, fill=blue!60!red] (5) at (axis cs:476, 398) {\lab{5}};
\node[dot, fill=blue!60!red] (6) at (axis cs:396, 409) {\lab{6}};
\node[dot, fill=blue!80!red] (7) at (axis cs:379, 257) {\lab{7}};
\node[dot, fill=blue!100!red] (8) at (axis cs:348, 289) {\lab{8}};
\node[dot, fill=blue!40!red] (9) at (axis cs:325, 451) {\lab{9}};
\node[dot, fill=blue!20!red] (10) at (axis cs:379, 68) {\lab{10}};
\node[dot, fill=blue!80!red] (11) at (axis cs:335, 364) {\lab{11}};
\node[dot, fill=blue!60!red] (12) at (axis cs:367, 120) {\lab{12}};
\node[dot, fill=blue!80!red] (13) at (axis cs:344, 227) {\lab{13}};
\node[dot, fill=blue!40!red] (14) at (axis cs:358, 107) {\lab{14}};
\node[dot, fill=blue!20!red] (15) at (axis cs:359, 93) {\lab{15}};
\node[dot, fill=blue!80!red] (16) at (axis cs:268, 325) {\lab{16}};
\node[dot, fill=blue!80!red] (17) at (axis cs:263, 338) {\lab{17}};
\node[dot, fill=blue!60!red] (18) at (axis cs:317, 192) {\lab{18}};
\node[dot, fill=blue!60!red] (19) at (axis cs:234, 412) {\lab{19}};
\node[dot, fill=blue!60!red] (20) at (axis cs:174, 408) {\lab{20}};
\node[dot, fill=blue!40!red] (21) at (axis cs:143, 409) {\lab{21}};
\node[dot, fill=blue!20!red] (22) at (axis cs:94, 479) {\lab{22}};
\node[dot, fill=blue!0!red] (23) at (axis cs:0, 501) {\lab{23}};

    \end{axis}

    \node (coordlist) at (current axis.east)
        [anchor=west,
         xshift=1cm,
         font=\small]
    {
        \begin{tabular}{@{}l@{}}
        (\phantom{00}0, 501) \\
        (\phantom{0}94, 479) \\
        (143, 409) \\
        (174, 408) \\
        (234, 412) \\
        (263, 338) \\
        (268, 325) \\
        (317, 192) \\
        (325, 451) \\
        (335, 364) \\
        (344, 227) \\
        (348, 289) \\
        (358, 107) \\
        (359, \phantom{0}93) \\
        (367, 120) \\
        (379, \phantom{0}68) \\
        (379, 257) \\
        (389, \phantom{00}0) \\
        (396, 409) \\
        (476, 398) \\
        (579, 426) \\
        (648, 541) \\
        (672, 566)
        \end{tabular}
    };
  \end{tikzpicture}
    \caption{A set of 23 points with no 6-hole or 7-gon}
    \label{fig:23}
\end{figure}

At a high level, our approach involves encoding an underapproximation of a given geometric problem into SAT, which means that every solution to the underlying geometric problem corresponds to a satisfying assignment but not vice versa. A satisfying assignment is called an \emph{abstract order type}, and we say that an assignment corresponding to a geometric solution is \emph{realizable}. There are two challenges to overcome:
\begin{enumerate}
    \item Testing whether an abstract order type is realizable is computationally difficult (specifically, $\exists \mathbb{R}$-complete).
    \item Almost all abstract order types are not realizable.
\end{enumerate}

Subercaseaux, Mackey, Qian, and Heule \cite{symmetric} recently made significant progress on the first challenge by developing \textsf{Localizer}, a local-search solver for testing the realizability of abstract order types. Given an abstract order type as input, \textsf{Localizer} initializes a point configuration and iteratively perturbs the points using a simulated annealing strategy to minimize the number of violated orientation constraints. For realizable order types, \textsf{Localizer} succeeds in finding a realization with sufficiently high frequency to make it the current state-of-the-art tool for this task.

Our contribution is to the second challenge. If almost all abstract order types are not realizable, we need some way to guide our search toward the few that are. Indeed, while searching for a set of 23 points with no 6-hole or 7-gon, Heule and Scheucher~\cite{h6-upper-bound} report that none of the abstract order types they tested were realizable. \textsf{PointSAT} interfaces between a SAT solver and \textsf{Localizer} while incorporating heuristics we discovered that make finding realizable order types more likely. By analogy with solvers for satisfiability modulo theories (SMT), which interface between SAT solvers and theory solvers, we call our approach \emph{satisfiability modulo realizability}, and \textsf{PointSAT} can be seen as an SMT solver for a fragment of the existential theory of the reals, albeit one designed only for satisfiable instances.

Our paradigm differs from the previous SMT-style solvers in a few ways. First, rather than add unrealizability lemmas to the SAT instance, we instruct the SAT solver to generate a diverse set of candidate solutions to avoid getting stuck in problematic regions of the search space (see \Cref{sec-diverse}). Second, when the candidate solution is rejected by the ``theory solver'' (i.e., \textsf{Localizer}), the latter sends a nearby candidate solution to the SAT solver, which often turns out to be a solution (see \Cref{sec-partial-real}). Third, we use a novel ``flippability'' heuristic, in which only partial information from a candidate solution is passed from the SAT solver to \textsf{Localizer} (see \Cref{sec-omit-flip}). This insight is based on a surprising experimental observation about solutions to geometric problems of the sort \textsf{PointSAT} applies to. All three of these heuristics are essential for \textsf{PointSAT}'s performance and capabilities.

SAT solvers have been successfully applied to many problems in pure mathematics (see, e.g., \cite{discrepancy,pythag,schur,keller,kaplansky,packing}), and they are now an indispensable tool for computer-assisted mathematical discovery. On the other hand, since SAT solvers are not expressive enough for many mathematical problems, further progress in computer-assisted mathematics could be spurred by efficient solvers for more expressive theories. By using an SMT-style solver to resolve an open problem in pure mathematics, our work demonstrates the potential of this direction.

The outline of this paper is as follows. In \Cref{sec-background}, we provide context for the problems we study and discuss related work. In \Cref{sec-order-types}, we give the necessary background on the concepts from discrete geometry we employ. In \Cref{sec-approach}, we describe how \textsf{PointSAT} works and discuss the heuristics we discovered that are essential to overcoming the second challenge mentioned above. In \Cref{sec-evaluation}, we evaluate the effect these heuristics have on the performance of the solver and describe how we found the point set depicted in \Cref{fig:23}. In \Cref{sec-results}, we apply \textsf{PointSAT} to three other problems in discrete geometry and present comparative data from these experiments. Finally, in \Cref{sec-conclusion}, we conclude and discuss potential future work.

\section{The Happy Ending Problem and Its Variants} \label{sec-background}

In one of the first papers on Ramsey theory, Erd\H{o}s and Szekeres~\cite{es} proved that for every $k \in \mathbb{N}$, there is a $g(k) \in \mathbb{N}$ such that every set of at least $g(k)$ points in $\mathbb{R}^2$ in general position contains a $k$-gon. In fact, they conjectured that $g(k) = 2^{k-2} + 1$; it is known that $g(k) \ge 2^{k-2} + 1$~\cite{es2}. This conjecture remains unproven, but it is known to be true for $k \le 6$; the case $k=6$ was proven using a computer by Szekeres and Peters~\cite{sp}, and SAT solvers can now prove it in seconds~\cite{h6-upper-bound}.

A variant of the problem posed by Erd\H{o}s~\cite{erdos-empty} asks about avoiding $k$-holes. Let $h(k)$ be the minimum integer such that every set of at least $h(k)$ points in general position contains a $k$-hole, or let $h(k) = \infty$ if there is no such integer. The values of $h(k)$ are now fully determined: $h(3) = 3$, $h(4) = 5$, $h(5) = 10$~\cite{h5}, $h(6) = 30$~\cite{h6-lower-bound,h6-upper-bound}, and $h(7) = \infty$~\cite{h7}. We can also ask about simultaneously avoiding $k$-holes and $\ell$-gons. Let $h(k,\ell)$ be the minimum integer such that every set of at least $h(k,\ell)$ points in general position contains a $k$-hole or an $\ell$-gon. Previous constructions show that $h(5,6) = h(5) = 10$ and $h(6,8) = h(6) = 30$, which leaves $h(6,7)$ as the only nontrivial remaining case for $k \le 6$. Heule and Scheucher~\cite{h6-upper-bound} proved that $h(6,7) \le 24$. \Cref{thm:23}, proven using \textsf{PointSAT}, states that $h(6,7) = 24$.

\begin{wrapfigure}[10]{r}{0.4\textwidth}
    \centering
    \begin{tikzpicture}[
    scale=0.67,
    xscale=0.8,
    yscale=0.7,
    cap path/.style={thick, blue, join=round},
    extended line/.style={dashed, gray, thick},
    drop line/.style={dotted, red, thick, ->, >=stealth},
    point/.style={circle, fill=black, inner sep=1.5pt}
]

    \coordinate (p1) at (0, 0);
    \coordinate (p2) at (2, 2);
    \coordinate (p3) at (4, 3);
    \coordinate (p4) at (6, 2.5);
    \coordinate (p5) at (8, 1);

    \draw[cap path] (p1) -- (p2) -- (p3) -- (p4) -- (p5);

    \node[point, label=below:$p_1$] at (p1) {};
    \node[point, label=below right:$p_2$] at (p2) {};
    \node[point, label=below:$p_3$] at (p3) {};
    \node[point, label=below:$p_4$] at (p4) {};
    \node[point, label=below:$p_5$] at (p5) {};

\end{tikzpicture}
    \caption{A 5-cap}
    \label{fig:placeholder}
\end{wrapfigure}

Let a \emph{$k$-cap} be a sequence of points $p_1,\dots,p_k$ with increasing $x$-coordinates such that $p_{i+2}$ lies below the line from $p_i$ to $p_{i+1}$ for all $i \in [k-2]$. \textsf{PointSAT} can find a set of 26 points with no 5-cap or 7-gon. The existence of such a set follows from a theorem of Erd\H{o}s, Tuza, and Valtr~\cite{etv}, although \textsf{PointSAT} works without relying on any problem-specific knowledge. \textsf{PointSAT} can also find a set of 29 points with no 6-hole, replicating a result of Overmars~\cite{h6-lower-bound}. Overmars wrote a dedicated program to search for such a set of points, which relies on a specialized algorithm for finding holes~\cite{empty-alg};\footnote{Overmars' program was publicly released, but it appears to now be lost unfortunately.} \textsf{PointSAT} uses no problem-specific algorithms.

\section{Order Types and Realizability} \label{sec-order-types}

Problems in discrete geometry are naturally formulated as satisfiability problems with constraints involving real numbers representing the coordinates of points. However, this formulation obscures the combinatorial content of the problems, and applying an off-the-shelf solver for the existential theory of the reals to such an encoding would be hopelessly inefficient. We instead adopt a different approach based on \emph{order types}~\cite{knuth}, which is standard when applying SAT solvers to problems in discrete geometry (see, e.g., \cite{h6-upper-bound,symmetric}).

We consider problems concerning points in $\mathbb{R}^2$ in general position. Recall that this means that no three points are collinear. Given such a set of $n$ points, we can rotate the set so that no two points share the same $x$-coordinate. Then, we can label the points from left to right with labels $1,\dots,n$. For each triple of points $(i,j,k)$ with $1 \le i < j < k \le n$, let $\sigma(i,j,k) = +$ if $i$, $j$, and $k$ are oriented counterclockwise and $\sigma(i,j,k) = -$ if they are oriented clockwise; since the points are in general position, $i$, $j$, and $k$ must be oriented either clockwise or counterclockwise. This function $\sigma$ is called the \emph{order type} of the set of points, and a constraint of the form $\sigma(i,j,k) = \pm$ is called an \emph{orientation constraint}. By considering order types rather than concrete sets of points in $\mathbb{R}^2$, we lose information about the exact positioning of the points, but we retain exactly the information we need to define the geometric notions we care about, like convexity.

Not every function $\sigma : \binom{[n]}{3} \to \{-,+\}$ corresponds to a geometric configuration of points;\footnote{By $\binom{[n]}{3}$, we mean the set of all 3-element subsets of $\{1,\dots,n\}$.} we say that such a function is a \emph{realizable order type} if it is induced by a set of points in $\mathbb{R}^2$. Given points $i$, $j$, $k$, and $\ell$ such that $1 \le i < j < k < \ell \le n$, it is not hard to see that the following constraints hold for any realizable order type:
\begin{align*}
    \sigma(i,j,k) = + \ \land \ \sigma(i,k,\ell) = + \quad &\rightarrow \quad \sigma(i,j,\ell) = + \\
    \sigma(i,j,k) = - \ \land \ \sigma(i,k,\ell) = - \quad &\rightarrow \quad \sigma(i,j,\ell) = - \\
    \sigma(i,j,k) = + \ \land \ \sigma(j,k,\ell) = + \quad &\rightarrow \quad \sigma(i,k,\ell) = + \\
    \sigma(i,j,k) = - \ \land \ \sigma(j,k,\ell) = - \quad &\rightarrow \quad \sigma(i,k,\ell) = -.
\end{align*}
We call these constraints the \emph{order type axioms},\footnote{They are also sometimes called the \emph{signotope axioms}~\cite{signotope}.} and we call a function $\sigma : \binom{[n]}{3} \to \{-,+\}$ satisfying these constraints an \emph{abstract order type}. The order type axioms can be naturally encoded into SAT using a boolean variable for each $\sigma(i,j,k)$ (called an \emph{orientation variable}), making them particularly suitable for our purposes. However, not every abstract order type is realizable, and it is $\exists \mathbb{R}$-complete to determine whether an abstract order type is realizable~\cite{mnev,shor}. Therefore, a compact SAT encoding of the exact constraints for realizable order types would imply that $\mathsf{NP} = \exists \mathbb{R}$, which is considered unlikely. We don't attempt to enforce any additional realizability constraints at the SAT level.\footnote{One justification for this is that the smallest unrealizable abstract order types have 9 points~\cite{realizable8}, and forbidding these sub-configurations would seem to require $O(n^9)$ clauses, which is quite expensive.}

The papers \cite{h6-upper-bound,symmetric} describe how various geometric constraints can be expressed in terms of abstract order types and encoded into SAT. Specifically, \cite{h6-upper-bound} shows how to encode avoiding $k$-holes, $\ell$-gons, and $k$-caps, and \cite{symmetric} shows how to enforce rotational symmetry and prescribe the sizes of the convex hull layers. For all of the solvable problems we consider in this paper, the SAT instances can be solved in minutes on a single core. The difficulty is in finding realizable solutions.

Despite the order type axioms being an underapproximation of the geometric constraints, the remarkable fact is that all of the theorems mentioned in \Cref{sec-background} remain true when they are instead interpreted as statements about abstract order types. On the other hand, almost all abstract order types are not realizable: There are $2^{\Theta(n^2)}$ abstract order types on $n$ points, but only $2^{\Theta(n \log n)}$ of them are realizable~\cite{nlogn}. In the context of a geometric problem that can be expressed as a constraint satisfaction problem on abstract order types, we say that an abstract order type that satisfies the constraints is an \emph{abstract solution}; if an abstract solution is realizable, we say it is a \emph{solution}.

\section{Description of \textsf{PointSAT}} \label{sec-approach}
In this section, we describe the operation of \textsf{PointSAT}. We begin with an outline of the methodology, and then we describe three techniques we employed that are essential for making the approach feasible and performant.

At a high level, the basic strategy is as follows:
\begin{enumerate}
    \item Encode the geometric problem into a SAT instance.
    \item Find many abstract solutions using a SAT solver. 
    \item For each abstract solution, use \textsf{Localizer} to search for a realization.
\end{enumerate}
This is effectively the same approach attempted by Heule and Scheucher~\cite{h6-upper-bound}, with the only difference being that they used a different tool for testing whether an abstract solution is realizable.\footnote{This tool was developed by Scheucher and is not publicly available.} Using this basic approach, we tried and failed to realize over 37,000 abstract solutions, replicating Heule and Scheucher's attempt. It is unsurprising that this strategy fails, because almost all abstract solutions are not realizable. The three techniques we describe next are designed to address this problem.

As will be shown in \Cref{sec-evaluation}, not all of these techniques are necessary to find a solution to prove \Cref{thm:23}. However, the combination of all of our techniques dramatically improves \textsf{PointSAT}'s performance, which allowed us to find over a thousand solutions proving \Cref{thm:23}. We observed that many solutions found by \textsf{PointSAT} have points tightly clustered together, so much so that often only a few clusters of points can be visually distinguished.\footnote{There are theoretical results to the effect that, in the worst case, realizations of order types may require the points to be tightly clustered together~\cite{spread}.} Therefore, to find a solution like the one in \Cref{fig:23}, where every point can be individually distinguished, it was necessary to generate many solutions. The performance gains are also useful for applying \textsf{PointSAT} to harder problems.

\subsection{Generating Diverse Abstract Solutions} \label{sec-diverse}
A common approach for enumerating solutions to a SAT formula is to use a tool like \textsf{allsat-CaDiCaL}~\cite{allsat}. For SAT formulas with few satisfying assignments, such tools are an effective way to enumerate them all. However, for the problems we consider, the number of abstract solutions is too large to exhaustively enumerate. While we can use \textsf{allsat-CaDiCaL} to generate thousands or millions of abstract solutions, most of these abstract solutions will be very similar to each other (as measured by, e.g., Hamming distance). Thus, most of the generated abstract solutions will likely be unrealizable for the same reason.

To address this problem, we generate a \emph{diverse} set of abstract solutions. Previous researchers have considered the problem of finding a diverse set of satisfying assignments of a SAT formula~\cite{uniform-sampling,diverse-sat,exact-diverse-sat}, but we opt for a very simple approach that suffices for our purpose. To generate an abstract solution of a formula, we use the \textsf{scranfilize} tool~\cite{scranfilize} to randomly permute the clauses of the formula, and then we use a SAT solver (such as \textsf{Kissat}) to solve the resulting formula. The effect of permuting the clauses is that the SAT solver finds a different solution each time. The resulting solutions are diverse enough to avoid the above problem.

\subsection{Testing if Partial Realizations Are Solutions} \label{sec-partial-real}
When \textsf{Localizer} fails to find a realization of a given set of constraints, it outputs the best point set it found, meaning the point set with the fewest number of violated constraints. We call such a point set a \emph{partial realization}. Partial realizations sometimes end up being solutions anyway; that is, despite not agreeing with the abstract solution we were targeting, a partial realization may still satisfy all of the desired combinatorial constraints. Thus, when \textsf{Localizer} outputs a partial realization, we always check whether it is a solution. We do this by calculating the orientations from the partial realization, adding the corresponding unit clauses to the SAT formula, and using a SAT solver to check whether these orientations indeed satisfy the combinatorial constraints.

\subsection{Omitting Flippable Orientations} \label{sec-omit-flip}
Given a variable assignment $\alpha : X \to \{\bot,\top\}$ and a variable $x \in X$, let $\alpha_x$ be the assignment obtained by flipping the assignment to variable $x$; that is, $\alpha_x(x) = \overline{\alpha(x)}$ and $\alpha_x(y) = \alpha(y)$ for all $y \in X \setminus \{x\}$. Given a SAT formula $\varphi$ with variables $X$ and a satisfying assignment $\alpha$, we say that a variable $x \in X$ is \emph{flippable} if $\alpha_x$ is a satisfying assignment. While analyzing the structure of abstract solutions to geometric problems, we noticed that a remarkable number of orientation variables are flippable. For example, the abstract solutions of 23 points avoiding 6-holes and 7-gons that we found have about 35.9 flippable orientation variables on average (see \Cref{fig:flippables-23}).

Motivated by this observation, we came up with the following strategy. Instead of passing a complete abstract order type to \textsf{Localizer}, we omit all of the orientation constraints corresponding to flippable variables. Naturally, after dropping these constraints, the remaining constraints are more likely to be realizable, and \textsf{Localizer} is more likely to find a realization (or find a partial realization satisfying more constraints). The goal is to remove constraints that are less critical for being a solution. The trade-off is that if \textsf{Localizer} finds a realization, it may no longer be a solution; flipping the assignment to multiple flippable variables does not necessarily result in a satisfying assignment. But, perhaps surprisingly, this happens much less often than one may expect. For example, applying this strategy to the search for a set of 23 points avoiding 6-holes and 7-gons, we found 90 realizations satisfying the non-flippable orientations of some abstract solutions, and 81 of these were solutions (see \Cref{fig:prob-23}).

One possible explanation for why omitting flippable orientations works so well is the following. Since almost all abstract order types are non-realizable, it is almost certain that \textsf{Localizer} is going to violate some of the orientation constraints. But not all orientation constraints are of equal importance for ensuring that the partial realization is a solution: flippable orientations are less important than non-flippable orientations. By omitting the flippable orientations entirely, we incentivize \textsf{Localizer} to violate flippable orientations rather than non-flippable orientations, making the resulting partial realizations more likely to be solutions.

\subsection{Summary}
In summary, our revised strategy for finding a solution to a geometric problem is as follows:
\begin{enumerate}
    \item Encode the geometric problem into a SAT instance.
    \item Find a diverse set of abstract solutions by repeatedly applying \textsf{scranfilize} to the formula and solving the resulting scrambled formula using a SAT solver.
    \item For each abstract solution, identify the flippable orientations and remove them to create a partial abstract solution.
    \item For each partial abstract solution, run \textsf{Localizer} with some timeout and obtain a partial realization.
    \item For each partial realization, test whether it is a solution.
\end{enumerate}
\noindent
Step 1 is done by the user, and steps 2--5 are performed by \textsf{PointSAT}.

\section{Experiments for 23 Points With No 6-Gon or 7-Hole} \label{sec-evaluation}

\subsection{Evaluation}

In this section, we evaluate the effect of the heuristics discussed in \Cref{sec-approach}, using the problem of finding a set of 23 points avoiding 6-holes and 7-gons as a benchmark. In \Cref{sec-results}, we will demonstrate that our approach generalizes to other problems in discrete geometry as well.

To evaluate the effectiveness of our heuristics, we generated 6494 abstract solutions for this problem using \textsf{CaDiCaL}. 
Using these abstract solutions, we evaluated the impact of testing partial realizations (\Cref{sec-partial-real}) and omitting flippable orientations (\Cref{sec-omit-flip}), resulting in four different strategies depending on whether each heuristic is used.
In each case, \textsf{PointSAT} ran \textsf{Localizer} with a timeout of 15 seconds on a single thread for each abstract solution. The number of solutions obtained using each strategy is given in \Cref{tab:num-realizations-diff-settings}.

\begin{table}[ht]
    \centering
    \caption{The number of solutions found for each of the four strategies}
    \label{tab:num-realizations-diff-settings}
    \begin{tabular}{c@{~~~~~}c@{~~~~~}c}
    \toprule
          & Test partial realization & Do not test partial realization \\
          \midrule
          Omit flippable & {\bf 42} & 8\\
          Keep flippable   & 3 & 0\\
    \bottomrule
    \end{tabular}
\end{table}

Notice that without either of our heuristics, we do not find any solutions. Thus, the heuristics we developed were necessary for finding a set of 23 points with no 6-hole or 7-gon. Each heuristic from \Cref{sec-partial-real,sec-omit-flip} by itself allows us to find some solutions, but the results improve dramatically when they are combined.

In addition to these results, it is also insightful to consider the number of violated constraints \textsf{Localizer} was able to obtain both when flippables are kept and omitted. These data are shown in \Cref{fig:num-violations-compare}. As one can see, omitting flippable literals shifts the distribution of the number of violations to the left, which is to be expected since there are fewer constraints. Specifically, the mean number of violations is 20.1 when flippable orientations are kept and 9.7 when they are omitted.\footnote{When computing statistics like these, we remove data in the top percentile, because there are rare instances where \textsf{Localizer} outputs a point set with an exceptionally high number of violations.} This is part of the explanation for why the heuristic is effective. To be sure, when flippable orientations are omitted, many of these will be implicitly violated, which is not reflected in these counts and may contribute to a partial realization's failure to be a solution. Evidently, however, the trade-off works out such that the heuristic is effective, which can be attributed to the fact that satisfying the flippable orientations is less important than the non-flippable orientations.

\begin{figure}[ht]
    \centering
    \begin{tikzpicture}
        \pgfplotsset{
            histogram style/.style={
                width=0.8\linewidth,
                height=0.5\linewidth,
                ybar interval,
                x tick label as interval=false,
                ymin=0,
                ymax=0.12,
                xmin=0,
                xmax=36,
                xtick={0, 4, 8, 12, 16, 20, 24, 28, 32, 36},
                tick label style={font=\footnotesize},
                label style={font=\small},
                grid style={line width=.1pt, draw=gray!20}
            }
        }

        \begin{axis}[
            histogram style,
            xlabel={Number of violations},
            ylabel={Frequency},
            axis x line*=bottom,
            axis y line=none,
            ymajorgrids=true,
            xmajorgrids=false,
            ytick={0, 0.02, 0.04, 0.06, 0.08, 0.1},
            yticklabel=\empty,
            legend style={at={(0.98,0.98)}, anchor=north east,legend columns=1,legend cell align=left, font=\footnotesize, fill=none}
        ]
        
        \addlegendimage{area legend, fill=red!50, draw=red!70!black, fill opacity=0.5}
        \addlegendentry{Flippables omitted}

        \addlegendimage{area legend, fill=blue!50, draw=blue!70!black, fill opacity=0.5}
        \addlegendentry{Flippables kept}
        
        \addplot[fill=red!50, draw=red!70!black, fill opacity=0.5] coordinates {
            (0.0, 0.001397) (1.0, 0.006520) (2.0, 0.015368) (3.0, 0.030115) (4.0, 0.043310) 
            (5.0, 0.059919) (6.0, 0.081031) (7.0, 0.083514) (8.0, 0.091742) (9.0, 0.101366) 
            (10.0, 0.098572) (11.0, 0.075442) (12.0, 0.067370) (13.0, 0.059609) (14.0, 0.048432) 
            (15.0, 0.037100) (16.0, 0.030115) (17.0, 0.024216) (18.0, 0.017231) (19.0, 0.009624) 
            (20.0, 0.008382) (21.0, 0.006520) (22.0, 0.003105) (23.0, 0.002464) (24.0, 0.001232) 
            (25.0, 0.000924) (26.0, 0.000770) (27.0, 0.001232) (28.0, 0.000000) (29.0, 0.000000) 
            (30.0, 0.000308) (31.0, 0.000154) (32.0, 0.000154) (33.0, 0.000000) (34.0, 0.000000) 
            (35.0, 0.000000) (36.0, 0.000154) (37.0, 0.000000)
        };
        \end{axis}

        \begin{axis}[
            histogram style,
            axis lines=none,
            ticks=none,
            grid=none
        ]
        
        \addplot[fill=blue!50, draw=blue!70!black, fill opacity=0.5] coordinates {
            (4.0, 0.000310) (5.0, 0.000621) (6.0, 0.000776) (7.0, 0.001708) (8.0, 0.006675) 
            (9.0, 0.007917) (10.0, 0.011953) (11.0, 0.016610) (12.0, 0.024371) (13.0, 0.035548) 
            (14.0, 0.042068) (15.0, 0.055418) (16.0, 0.065663) (17.0, 0.075132) (18.0, 0.067215) 
            (19.0, 0.070475) (20.0, 0.073890) (21.0, 0.065042) (22.0, 0.067060) (23.0, 0.056659) 
            (24.0, 0.052313) (25.0, 0.042999) (26.0, 0.036014) (27.0, 0.026700) (28.0, 0.023440) 
            (29.0, 0.019249) (30.0, 0.013971) (31.0, 0.010866) (32.0, 0.009003) (33.0, 0.008693) 
            (34.0, 0.004657) (35.0, 0.004346) (36.0, 0.002639) (37.0, 0.000000)
        };
        \end{axis}

    \end{tikzpicture}
    \caption{Distribution of number of violations in the best partial realizations \textsf{Localizer} obtained both with and without flippable orientations}
    \label{fig:num-violations-compare}
\end{figure}
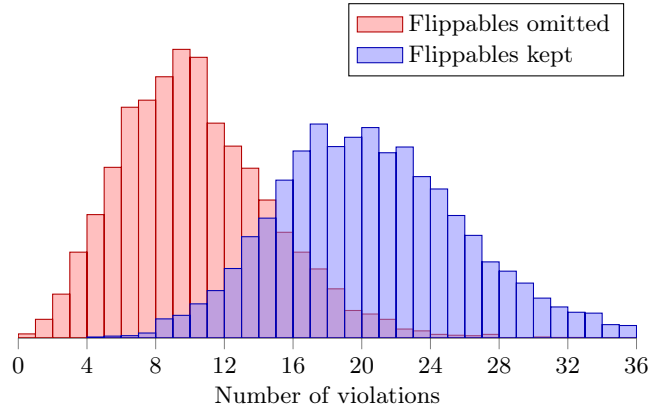

\subsection{Finding \Cref{fig:23}} \label{sec-finding}

Both for the sake of human-interpretability and aesthetics, we wanted to find a solution in which each of the points can be individually distinguished. Doing this required finding many solutions, selecting the best one, and applying post-processing to the point set.

We applied \textsf{PointSAT} to this problem by running it on the Pittsburgh Supercomputing Center~\cite{bridges2} on a node with 128 cores for a total of 2326 core hours. It generated 174,450 abstract solutions, from which it found 1035 solutions. For each abstract solution, \textsf{PointSAT} ran \textsf{Localizer} on two threads with a 15-second timeout.

We visually inspected each solution and found only one solution for which each point was clearly individually distinguishable. To simplify the representation of the solution, we wanted to transform the solution into one for which each coordinate is a small integer.\footnote{Minimizing the size of the integer grid into which a solution can be embedded is a task that has also been studied in the context of similar problems~\cite{horton-small,es-small}.} We started by dilating the point set by various factors and rounding the coordinates to integers, choosing the smallest dilation factor for which the resulting point set was still a solution. Then, we ran a script that uses simulating annealing to perturb the points with the goal of minimizing the smallest bounding box containing the points. Once this reached a local optimum, we ran a separate script that uses simulated annealing to minimize the ratio of the areas of the largest to smallest triangles. Finally, we ran the former script again, and the result was a solution that fits within the integer grid $\{0,\dots,672\} \times \{0,\dots,566\}$. This is the solution depicted in \Cref{fig:23}.

\subsection{Convex Hull Layers}

Among the 174,450 abstract solutions, there were 29 different possibilities for the convex hull layers, and among the 1035 solutions, there were 17 different possibilities (see \Cref{tab:convex-hull-layers}); note that the convex hull layers of an abstract order type are well defined \cite{knuth}. It is interesting that (3, 4, 4, 6, 5, 1) is by far the most common; the point set in \Cref{fig:23} is of this type. Every abstract solution we found has 3 points in its convex hull. It turns out that this is necessary:
\begin{theorem}
    Every abstract order type of 23 points with no 6-hole or 7-gon has 3 points in its convex hull.
\end{theorem}
\begin{proof}
    Our encoding enforces a symmetry breaking constraint that $\sigma(1,i,j) = +$ for all $2 \le i < j \le 23$ (see \cite{h6-upper-bound} for the justification of this constraint). Given this constraint, if there is a solution with at least 4 points in the convex hull, then there is a solution in which $\sigma(2,i,23) = +$ for some $i \in [3,22]$, which we can enforce with a single additional clause. We partitioned the resulting formula into cubes using the strategy from \cite{h6-upper-bound}; that is, we case split on all the possible assignments to a small set of orientation variables and solve the cases in parallel. The computation took 1727 core hours parallelized across 128 cores on the Pittsburgh Supercomputing Center, and each cube was unsatisfiable.
\end{proof}

\begin{table}[ht]
    \centering
    \caption{Frequency of convex hull layer sizes for 23 points with no 6-hole or 7-gon}
    \label{tab:convex-hull-layers}
    \begin{tabular}{@{~~~}l@{~~~}r@{~~~}r@{~~~~~~~~~~~~}l@{~~~}r@{~~~}r@{~~~}}
    \toprule
    Hull layer sizes & \# abs. & \# sol. & Hull layer sizes & \# abs. & \# sol. \\
    \midrule
    (3, 3, 3, 3, 6, 5) & 2 & 0       & (3, 3, 5, 6, 5, 1) & 8975 & 51 \\
    (3, 3, 3, 4, 4, 5, 1) & 73 & 0   & (3, 4, 3, 4, 6, 3) & 621 & 0 \\
    (3, 3, 3, 4, 5, 4, 1) & 280 & 0  & (3, 4, 3, 5, 4, 4) & 742 & 2 \\
    (3, 3, 3, 4, 5, 5) & 128 & 0     & (3, 4, 3, 5, 5, 3) & 3576 & 1 \\
    (3, 3, 3, 5, 4, 5) & 479 & 0     & (3, 4, 3, 5, 6, 2) & 661 & 2 \\
    (3, 3, 3, 5, 5, 4) & 475 & 0     & (3, 4, 3, 6, 4, 3) & 6551 & 1 \\
    (3, 3, 3, 6, 4, 4) & 5000 & 0    & (3, 4, 4, 4, 6, 2) & 5240 & 30 \\
    (3, 3, 3, 6, 6, 2) & 134 & 0     & (3, 4, 4, 5, 4, 3) & 1600 & 4 \\
    (3, 3, 4, 4, 5, 4) & 2 & 0       & (3, 4, 4, 5, 5, 2) & 12631 & 26 \\
    (3, 3, 4, 4, 6, 3) & 242 & 1     & (3, 4, 4, 6, 5, 1) & 72706 & 556 \\
    (3, 3, 4, 5, 4, 4) & 581 & 0     & (3, 4, 5, 5, 5, 1) & 3364 & 28 \\
    (3, 3, 4, 5, 5, 3) & 1369 & 0    & (3, 4, 5, 6, 5) & 13368 & 61 \\
    (3, 3, 4, 5, 6, 2) & 2359 & 11   & (3, 5, 5, 5, 4, 1) & 6023 & 63 \\
    (3, 3, 4, 6, 4, 3) & 19807 & 40  & (3, 5, 6, 6, 3) & 6891 & 157 \\
    (3, 3, 4, 6, 6, 1) & 570 & 1     & & & \\ 
    \bottomrule
    \end{tabular}
\end{table}

\section{Comparison of Experiments for Four Problems} \label{sec-results}
In addition to proving \Cref{thm:23}, we applied \textsf{PointSAT} to three other constructive problems in discrete geometry:
\begin{enumerate}
    \item the existence of 26 points with no 5-cap or 7-gon;
    \item the existence of 29 points with no 6-hole; and
    \item the existence of 32 points with no 7-gon.
\end{enumerate}
Note that these problems were known to have solutions prior to the present paper. For each problem, we ran \textsf{PointSAT} on the Pittsburgh Supercomputing Center on a node with 128 cores. We configured \textsf{PointSAT} to run \textsf{Localizer} on one thread with a 15-second timeout for each abstract solution. To fairly compare the performance of \textsf{PointSAT} on the four problems, we decided to run \textsf{PointSAT} again on the 23 point problem, both because our experiment in \Cref{sec-finding} ran \textsf{Localizer} on two threads instead of one and because we had made some changes to \textsf{PointSAT} in the meantime (most significantly, we changed how \textsf{PointSAT} generates diverse abstract solutions).\footnote{In an earlier prototype of \textsf{PointSAT}, we generated diverse abstract solutions by partitioning the formula into cubes using the strategy from \cite{h6-upper-bound}. We tried to solve each cube using a SAT solver and kept the satisfying assignments from the solvable cubes. We later found that we could generate diverse abstract solutions more efficiently using \textsf{scranfilize}.}

For the 23 point problem, we ran \textsf{PointSAT} for 1811 core hours, and it generated 200,000 abstract solutions, from which it found 423 solutions. For the 26 point problem, we ran \textsf{PointSAT} for 1193 core hours, and it generated 200,000 abstract solutions, from which it found 32 solutions. For the 29 point problem, we ran \textsf{PointSAT} for 3433 core hours, and it generated 18,806 abstract solutions, from which it found 4 solutions. For the 32 point problem, we ran \textsf{PointSAT} for 2191 core hours, and it generated 200,000 abstract solutions but did not find any solutions. As expected, the problems become more difficult as the number of points increases, since a smaller fraction of abstract solutions are realizable.

In comparison with the other problems, finding an abstract solution to the 29 point problem is abnormally difficult; it takes on average 640 seconds to find an abstract solution with \textsf{Kissat}, while the other problems can all be solved in under 30 seconds. This is why we only generated 18,806 abstract solutions for this problem. On average, it took \textsf{PointSAT} 858 core hours to find a solution to this problem, since we found 4 solutions with 3433 core hours. Overmars~\cite{h6-lower-bound} reports that his dedicated program, which is no longer available, could find a solution in about 4 days on a Pentium III 500 MHz, so \textsf{PointSAT} is not competitive with his program, although it is more generally applicable. This suggests that further improvements to \textsf{PointSAT} are possible.

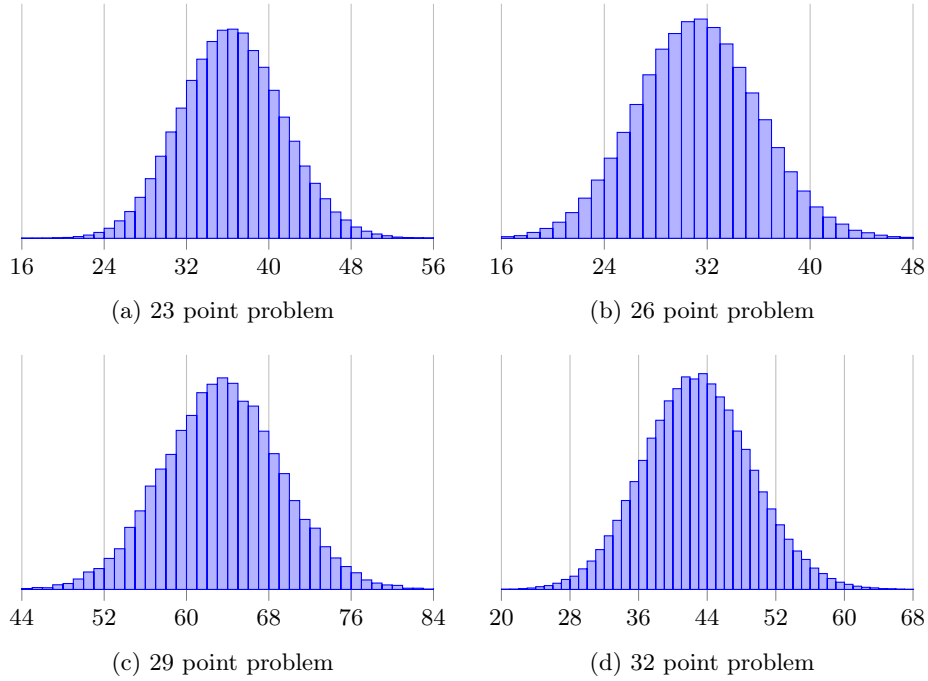
\begin{figure}[ht]
    \centering
    \begin{subfigure}[b]{0.48\linewidth}
        \centering
        \begin{tikzpicture}
\begin{axis}[
    width=1.2\linewidth,
    height=0.8\linewidth,
    ybar interval, 
    x tick label as interval=false,
    ymin=0,
    ymax=0.09, 
    xmin=16,
    xmax=56,
    xtick={16,24,32,40,48,56},
    ytick=\empty,
    grid=major,
    axis lines*=none,
    axis y line=none,
    fill=blue!40,
    draw=blue!60!black
]
\addplot coordinates {
    (15, 0.000005) (16, 0.000015) (17, 0.000020) (18, 0.000035) (19, 0.000160) (20, 0.000330)
    (21, 0.000665) (22, 0.001240) (23, 0.002255) (24, 0.003850) (25, 0.006690) (26, 0.010261)
    (27, 0.015711) (28, 0.022882) (29, 0.031490) (30, 0.040795) (31, 0.049958) (32, 0.060577)
    (33, 0.068725) (34, 0.075510) (35, 0.079835) (36, 0.080245) (37, 0.078753) (38, 0.072050)
    (39, 0.065607) (40, 0.056776) (41, 0.046491) (42, 0.037245) (43, 0.027732) (44, 0.021089)
    (45, 0.015176) (46, 0.010091) (47, 0.007025) (48, 0.004330) (49, 0.002725) (50, 0.001645)
    (51, 0.000990) (52, 0.000470) (53, 0.000265) (54, 0.000145) (55, 0.000100) (56, 0.000020)
    (57, 0.000005) (58, 0.000005) (59, 0.000005) (60, 0.000005)
};
\end{axis}
\end{tikzpicture}
        \caption{23 point problem}
        \label{fig:flippables-23}
    \end{subfigure}
    \hfill
    \begin{subfigure}[b]{0.48\linewidth}
        \centering
        \begin{tikzpicture}
\begin{axis}[
    width=1.2\linewidth,
    height=0.8\linewidth,
    ybar interval, 
    x tick label as interval=false,
    ymin=0,
    ymax=0.09,
    xmin=16,
    xmax=48,
    xtick={16,24,32,40,48},
    ytick=\empty,
    grid=major,
    axis lines*=none,
    axis y line=none,
    fill=blue!40,
    draw=blue!60!black
]
\addplot coordinates {
    (12, 0.000010) (13, 0.000065) (14, 0.000105) (15, 0.000270) (16, 0.000625) (17, 0.001020) 
(18, 0.002250) (19, 0.003780) (20, 0.006175) (21, 0.009865) (22, 0.015410) (23, 0.022350) 
(24, 0.030745) (25, 0.040635) (26, 0.051345) (27, 0.062735) (28, 0.072570) (29, 0.078570) 
(30, 0.083120) (31, 0.084125) (32, 0.080880) (33, 0.074570) (34, 0.065575) (35, 0.055770) 
(36, 0.045495) (37, 0.034800) (38, 0.025525) (39, 0.018140) (40, 0.012035) (41, 0.008115) 
(42, 0.005445) (43, 0.003080) (44, 0.002125) (45, 0.001200) (46, 0.000640) (47, 0.000375) 
(48, 0.000235) (49, 0.000140) (50, 0.000065) (51, 0.000005) (52, 0.000005) (53, 0.000000) 
(54, 0.000010) (55, 0.000000)
};
\end{axis}
\end{tikzpicture}
        \caption{26 point problem}
        \label{fig:flippables-26}
    \end{subfigure}

    \vspace{1em} 

    \begin{subfigure}[t]{0.48\linewidth}
        \centering
        \begin{tikzpicture}
\begin{axis}[
    width=1.2\linewidth,
    height=0.8\linewidth,
    ybar interval, 
    x tick label as interval=false,
    ymin=0,
    ymax=0.08,
    xmin=44,
    xmax=84,
    xtick={44,52,60,68,76,84},
    ytick=\empty,
    grid=major,
    axis lines*=none,
    axis y line=none,
    fill=blue!40,
    draw=blue!60!black
]
\addplot coordinates {
    (42, 0.000106) (43, 0.000213) (44, 0.000319) (45, 0.000638) (46, 0.000585) (47, 0.001595) 
(48, 0.002021) (49, 0.003563) (50, 0.005956) (51, 0.007125) (52, 0.010582) (53, 0.013879) 
(54, 0.021163) (55, 0.026800) (56, 0.035255) (57, 0.041051) (58, 0.046049) (59, 0.054238) 
(60, 0.059343) (61, 0.067106) (62, 0.069978) (63, 0.072158) (64, 0.070297) (65, 0.064182) 
(66, 0.062533) (67, 0.053972) (68, 0.046315) (69, 0.039509) (70, 0.030309) (71, 0.023875) 
(72, 0.020898) (73, 0.014570) (74, 0.010741) (75, 0.007817) (76, 0.005690) (77, 0.003244) 
(78, 0.002180) (79, 0.001542) (80, 0.001383) (81, 0.000425) (82, 0.000425) (83, 0.000106) 
(84, 0.000106) (85, 0.000106) (86, 0.000000) (87, 0.000053) (88, 0.000000)
};
\end{axis}
\end{tikzpicture}
        \caption{29 point problem}
        \label{fig:flippables-29}
    \end{subfigure}
    \hfill
    \begin{subfigure}[t]{0.48\linewidth}
        \centering
        \begin{tikzpicture}
\begin{axis}[
    width=1.2\linewidth,
    height=0.8\linewidth,
    ybar interval, 
    x tick label as interval=false,
    ymin=0,
    ymax=0.07,
    xmin=20,
    xmax=68,
    xtick={20,28,36,44,52,60,68},
    ytick=\empty,
    grid=major,
    axis lines*=none,
    axis y line=none,
    fill=blue!40,
    draw=blue!60!black
]
\addplot coordinates {
    (16, 0.000005) (17, 0.000010) (18, 0.000030) (19, 0.000025) (20, 0.000065) (21, 0.000085) 
(22, 0.000235) (23, 0.000435) (24, 0.000825) (25, 0.001150) (26, 0.001850) (27, 0.002975) 
(28, 0.003955) (29, 0.006140) (30, 0.008475) (31, 0.011895) (32, 0.016005) (33, 0.020620) 
(34, 0.026545) (35, 0.032188) (36, 0.038506) (37, 0.045151) (38, 0.050488) (39, 0.056271) 
(40, 0.059911) (41, 0.063428) (42, 0.062811) (43, 0.064396) (44, 0.061331) (45, 0.057446) 
(46, 0.053531) (47, 0.047266) (48, 0.041798) (49, 0.035561) (50, 0.029115) (51, 0.024025) 
(52, 0.019285) (53, 0.014980) (54, 0.011470) (55, 0.008955) (56, 0.006370) (57, 0.004745) 
(58, 0.003265) (59, 0.002160) (60, 0.001525) (61, 0.000970) (62, 0.000640) (63, 0.000475) 
(64, 0.000230) (65, 0.000185) (66, 0.000085) (67, 0.000030) (68, 0.000020) (69, 0.000025) 
(70, 0.000010) (71, 0.000005) (72, 0.000010) (73, 0.000010) (74, 0.000000)
};
\end{axis}
\end{tikzpicture}
        \caption{32 point problem}
        \label{fig:flippables-32}
    \end{subfigure}

    \caption{Number of flippable orientations}
    \label{fig:num-flippables}
\end{figure}

\subsection{Number of Flippable Orientations}

\Cref{fig:num-flippables} shows the distributions for the number of flippable orientations in abstract solutions for each of the four problems. Each distribution is roughly normally distributed with a mean of 35.9 for the 23 point problem (2.0\% of all triples), 30.7 for the 26 point problem (1.2\% of all triples), 63.0 for the 29 point problem (1.7\% of all triples), and 42.4 for the 32 point problem (0.9\% of all triples). The ratio of the number of flippable orientations to the number of triples of points was smaller for the 32 point problem than for the other three problems, which limits the effectiveness of omitting flippable orientations. This suggests that, besides having more points, the 32 point problem may have some additional problem-specific difficulty as well.

\subsection{Distribution of Number of Violations}

\Cref{fig:num-violations} shows the distributions for the number of violations for partial realizations (regardless of whether they are solutions) for each of the four problems. Each distribution is roughly normally distributed with a mean of 12.2 for the 23 point problem, 38.0 for the 26 point problem, 21.8 for the 29 point problem, and 121.6 for the 32 point problem. The data suggest that abstract solutions to the 26 and 32 point problems are more difficult to realize than abstract solutions for the 23 and 29 point problems.

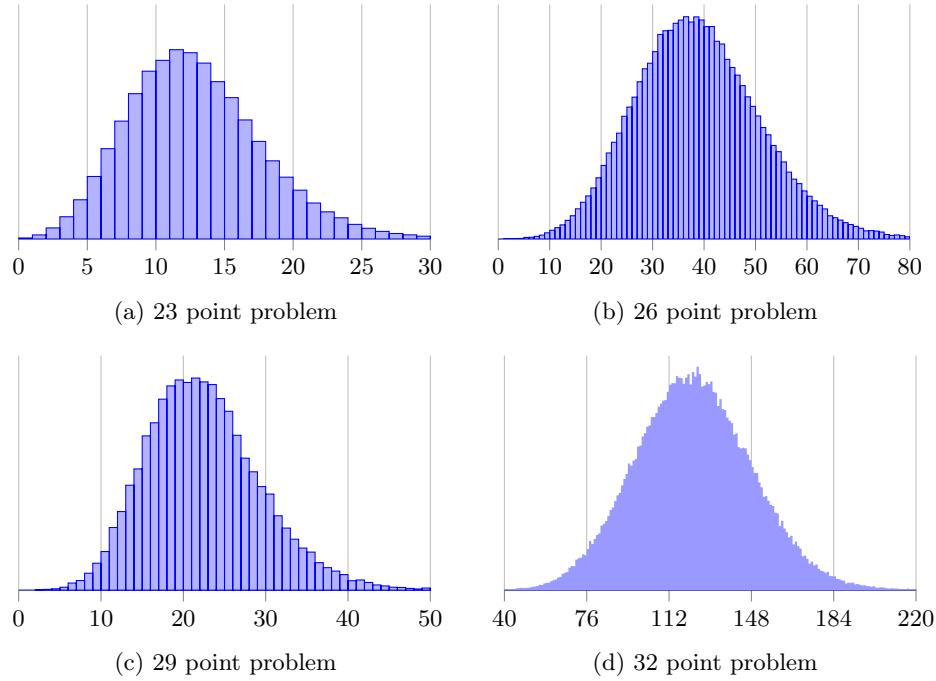
\begin{figure}[ht]
    \centering
    \begin{subfigure}[b]{0.48\linewidth}
        \centering
        \begin{tikzpicture}
\begin{axis}[
    width=1.2\linewidth,
    height=0.8\linewidth,
    ybar interval, 
    x tick label as interval=false,
    ymin=0,
    ymax=0.105,
    xmin=0,
    xmax=30,
    xtick={0,5,10,15,20,25,30},
    ytick=\empty,
    grid=major,
    axis lines*=none,
    axis y line=none,
    fill=blue!40,
    draw=blue!60!black
]
\addplot coordinates {
    (0, 0.000450) (1, 0.001815) (2, 0.004961) (3, 0.009961) (4, 0.017552) (5, 0.027984)
    (6, 0.040425) (7, 0.052767) (8, 0.065054) (9, 0.075110) (10, 0.079971) (11, 0.084736)
    (12, 0.083381) (13, 0.078736) (14, 0.070440) (15, 0.063284) (16, 0.053287) (17, 0.043676)
    (18, 0.034995) (19, 0.027674) (20, 0.021853) (21, 0.016112) (22, 0.012152) (23, 0.009486)
    (24, 0.006541) (25, 0.004791) (26, 0.003405) (27, 0.002460) (28, 0.001855) (29, 0.001255)
    (30, 0.000925) (31, 0.000615) (32, 0.000415) (33, 0.000295) (34, 0.000240) (35, 0.000225)
    (36, 0.000095) (37, 0.000095) (38, 0.000055) (39, 0.000015) (40, 0.000055) (41, 0.000030)
    (42, 0.000030) (43, 0.000015) (44, 0.000005) (45, 0.000015) (46, 0.000010) (47, 0.000010)
    (48, 0.000005) (49, 0.000005) (51, 0.000010) (52, 0.000005) (53, 0.000005) (55, 0.000005)
    (58, 0.000005) (60, 0.000005) (61, 0.000010) (64, 0.000005) (65, 0.000005) (87, 0.000005)
};
\end{axis}
\end{tikzpicture}
        \caption{23 point problem}
        \label{fig:violations-23}
    \end{subfigure}
    \hfill
    \begin{subfigure}[b]{0.48\linewidth}
        \centering
        \begin{tikzpicture}
\begin{axis}[
    width=1.2\linewidth,
    height=0.8\linewidth,
    ybar interval, 
    x tick label as interval=false,
    ymin=0,
    ymax=0.042,
    xmin=0,
    xmax=80,
    xtick={0,10,20,30,40,50,60,70,80},
    ytick=\empty,
    grid=major,
    axis lines*=none,
    axis y line=none,
    fill=blue!40,
    draw=blue!60!black
]
\addplot coordinates {
    (1, 0.000012) (2, 0.000047) (3, 0.000089) (4, 0.000101) (5, 0.000279) (6, 0.000314) 
(7, 0.000474) (8, 0.000688) (9, 0.001103) (10, 0.001559) (11, 0.002034) (12, 0.002526) 
(13, 0.003338) (14, 0.004162) (15, 0.005401) (16, 0.006575) (17, 0.007713) (18, 0.009142) 
(19, 0.010957) (20, 0.013168) (21, 0.015403) (22, 0.017318) (23, 0.018949) (24, 0.021972) 
(25, 0.023621) (26, 0.025393) (27, 0.028144) (28, 0.030474) (29, 0.031470) (30, 0.033545) 
(31, 0.036622) (32, 0.037286) (33, 0.037358) (34, 0.038626) (35, 0.039196) (36, 0.039800) 
(37, 0.038858) (38, 0.039771) (39, 0.039213) (40, 0.038028) (41, 0.038039) (42, 0.035401) 
(43, 0.034909) (44, 0.032170) (45, 0.030711) (46, 0.029342) (47, 0.027468) (48, 0.025458) 
(49, 0.023739) (50, 0.022032) (51, 0.019962) (52, 0.018652) (53, 0.016749) (54, 0.015640) 
(55, 0.013796) (56, 0.011970) (57, 0.010713) (58, 0.009919) (59, 0.008603) (60, 0.007654) 
(61, 0.006830) (62, 0.006030) (63, 0.005383) (64, 0.004156) (65, 0.003806) (66, 0.003385) 
(67, 0.002929) (68, 0.002591) (69, 0.002253) (70, 0.001921) (71, 0.001453) (72, 0.001500) 
(73, 0.001411) (74, 0.001245) (75, 0.000806) (76, 0.000634) (77, 0.000729) (78, 0.000593) 
(79, 0.000391) (80, 0.000457) (81, 0.000314)
};
\end{axis}
\end{tikzpicture}
        \caption{26 point problem}
        \label{fig:violations-26}
    \end{subfigure}

    \vspace{1em} 

    \begin{subfigure}[t]{0.48\linewidth}
        \centering
        \begin{tikzpicture}
\begin{axis}[
    width=1.2\linewidth,
    height=0.8\linewidth,
    ybar interval, 
    x tick label as interval=false,
    ymin=0,
    ymax=0.065,
    xmin=0,
    xmax=50,
    xtick={0,10,20,30,40,50,60,70,80},
    ytick=\empty,
    grid=major,
    axis lines*=none,
    axis y line=none,
    fill=blue!40,
    draw=blue!60!black
]
\addplot coordinates {
    (2, 0.000106) (3, 0.000213) (4, 0.000373) (5, 0.000799) (6, 0.001970) (7, 0.002662) 
(8, 0.004685) (9, 0.007453) (10, 0.010700) (11, 0.017354) (12, 0.021773) (13, 0.029066) 
(14, 0.033591) (15, 0.042747) (16, 0.046367) (17, 0.052915) (18, 0.056588) (19, 0.058238) 
(20, 0.057546) (21, 0.058770) (22, 0.057706) (23, 0.056960) (24, 0.053074) (25, 0.049561) 
(26, 0.042907) (27, 0.036572) (28, 0.032792) (29, 0.028693) (30, 0.026617) (31, 0.020602) 
(32, 0.017195) (33, 0.014107) (34, 0.011658) (35, 0.010540) (36, 0.007612) (37, 0.006228) 
(38, 0.005270) (39, 0.004312) (40, 0.002608) (41, 0.002821) (42, 0.002129) (43, 0.001437) 
(44, 0.001171) (45, 0.000799) (46, 0.000745) (47, 0.000479) (48, 0.000319) (49, 0.000586) 
(50, 0.000426) (51, 0.000160) (52, 0.000373) (53, 0.000213) (54, 0.000000) (55, 0.000106) 
(56, 0.000000) (57, 0.000053) (58, 0.000000) (59, 0.000053) (60, 0.000000)
};
\end{axis}
\end{tikzpicture}
        \caption{29 point problem}
        \label{fig:violations-29}
    \end{subfigure}
    \hfill
    \begin{subfigure}[t]{0.48\linewidth}
        \centering
        \begin{tikzpicture}
\begin{axis}[
    width=1.2\linewidth,
    height=0.8\linewidth,
    ybar interval, 
    x tick label as interval=false,
    ymin=0,
    ymax=0.47,
    xmin=40,
    xmax=220,
    xtick={40,76,112,148,184,220},
    ytick=\empty,
    grid=major,
    axis lines*=none,
    axis y line=none,
    fill=blue!40,
    draw=blue!60!black
]
\addplot+[fill=blue!40,draw=none] coordinates {
    (17, 0.000137) (18, 0.000000) (19, 0.000000) (20, 0.000137) (21, 0.000000) (22, 0.000000) 
(23, 0.000000) (24, 0.000000) (25, 0.000000) (26, 0.000000) (27, 0.000000) (28, 0.000000) 
(29, 0.000137) (30, 0.000000) (31, 0.000274) (32, 0.000411) (33, 0.000137) (34, 0.000137) 
(35, 0.000547) (36, 0.000137) (37, 0.000684) (38, 0.001232) (39, 0.001368) (40, 0.001368) 
(41, 0.001232) (42, 0.001505) (43, 0.001779) (44, 0.002189) (45, 0.003831) (46, 0.003968) 
(47, 0.004789) (48, 0.004105) (49, 0.004516) (50, 0.004789) (51, 0.005747) (52, 0.006705) 
(53, 0.007252) (54, 0.008347) (55, 0.010673) (56, 0.010400) (57, 0.013273) (58, 0.014778) 
(59, 0.014778) (60, 0.017926) (61, 0.019841) (62, 0.021483) (63, 0.025041) (64, 0.026546) 
(65, 0.032704) (66, 0.034756) (67, 0.036262) (68, 0.045019) (69, 0.048303) (70, 0.047893) 
(71, 0.061166) (72, 0.063903) (73, 0.065134) (74, 0.071839) (75, 0.069923) (76, 0.082375) 
(77, 0.098522) (78, 0.092365) (79, 0.102354) (80, 0.110290) (81, 0.113574) (82, 0.129310) 
(83, 0.132594) (84, 0.138342) (85, 0.148331) (86, 0.160509) (87, 0.165572) (88, 0.178298) 
(89, 0.189929) (90, 0.194718) (91, 0.209770) (92, 0.222906) (93, 0.233032) (94, 0.253695) 
(95, 0.249316) (96, 0.259579) (97, 0.262589) (98, 0.287767) (99, 0.289683) (100, 0.308019) 
(101, 0.314587) (102, 0.327586) (103, 0.330049) (104, 0.346606) (105, 0.346880) (106, 0.354817) 
(107, 0.367542) (108, 0.373700) (109, 0.379721) (110, 0.380952) (111, 0.397373) (112, 0.411604) 
(113, 0.414888) (114, 0.414477) (115, 0.428845) (116, 0.425698) (117, 0.424603) (118, 0.413519) 
(119, 0.425835) (120, 0.432677) (121, 0.412288) (122, 0.436918) (123, 0.427887) (124, 0.447729) 
(125, 0.432950) (126, 0.407909) (127, 0.410099) (128, 0.417214) (129, 0.412698) (130, 0.418172) 
(131, 0.397236) (132, 0.395594) (133, 0.368911) (134, 0.382731) (135, 0.362616) (136, 0.360974) 
(137, 0.345785) (138, 0.339491) (139, 0.334702) (140, 0.311713) (141, 0.314039) (142, 0.286535) 
(143, 0.284072) (144, 0.285167) (145, 0.284072) (146, 0.271757) (147, 0.260400) (148, 0.246716) 
(149, 0.234948) (150, 0.234674) (151, 0.215654) (152, 0.206076) (153, 0.187192) (154, 0.188424) 
(155, 0.178982) (156, 0.171045) (157, 0.168719) (158, 0.155309) (159, 0.153941) (160, 0.128352) 
(161, 0.138889) (162, 0.123563) (163, 0.116037) (164, 0.105638) (165, 0.108922) (166, 0.091680) 
(167, 0.097291) (168, 0.088123) (169, 0.083333) (170, 0.067871) (171, 0.069376) (172, 0.062397) 
(173, 0.053229) (174, 0.057471) (175, 0.049398) (176, 0.051450) (177, 0.042693) (178, 0.039272) 
(179, 0.037630) (180, 0.035030) (181, 0.036398) (182, 0.028188) (183, 0.027915) (184, 0.027641) 
(185, 0.022578) (186, 0.021894) (187, 0.019704) (188, 0.018610) (189, 0.013410) (190, 0.015052) 
(191, 0.014368) (192, 0.015189) (193, 0.011494) (194, 0.012315) (195, 0.010673) (196, 0.009442) 
(197, 0.009852) (198, 0.008210) (199, 0.007937) (200, 0.006842) (201, 0.007800) (202, 0.005337) 
(203, 0.004789) (204, 0.005473) (205, 0.005473) (206, 0.004516) (207, 0.003558) (208, 0.003284) 
(209, 0.003831) (210, 0.003284) (211, 0.003284) (212, 0.002737) (213, 0.002189) (214, 0.001232) 
(215, 0.001095) (216, 0.003147) (217, 0.001779) (218, 0.001368) (219, 0.000958) (220, 0.001505) 
(221, 0.001232) (222, 0.001095) (223, 0.000958) (224, 0.000958) (225, 0.000821) (226, 0.000684) 
(227, 0.000821) (228, 0.000821) (229, 0.000547) (230, 0.000411) (231, 0.000821) (232, 0.000274) 
(233, 0.000547) (234, 0.000274) (235, 0.000411) (236, 0.000684) (237, 0.000547) (238, 0.000684) 
(239, 0.000137) (240, 0.000137) (241, 0.000137) (242, 0.000274) (243, 0.000411) (244, 0.000137) 
(245, 0.000000) (246, 0.000000) (247, 0.000137) (248, 0.000137) (249, 0.000000) (250, 0.000137) 
(251, 0.000000)
};
\end{axis}
\end{tikzpicture}
        \caption{32 point problem}
        \label{fig:violations-32}
    \end{subfigure}

    \caption{Number of violations}
    \label{fig:num-violations}
\end{figure}

\subsection{Solution With a Given Number of Violations}

\Cref{fig:sols-violations-prob,fig:sols-violations-prop} show the probability of being a solution versus the number of violations and the proportion of solutions with a given number of violations. Note that we only present the data for the 23 and 26 point problems, because the other problems had too few or no solutions. The data for the 26 point problem are noisy, because there are only 32 solutions, but some patterns are still evident. There is a roughly exponential relationship between the number of violations in a partial realization and the probability that it is a solution. For the 23 point problem, the modal number of violations for the partial realizations that were solutions is 1, and for the 26 point problem, the mode is 9. This demonstrates the importance of checking whether a partial realization is a solution even when the number of violations is nonzero.

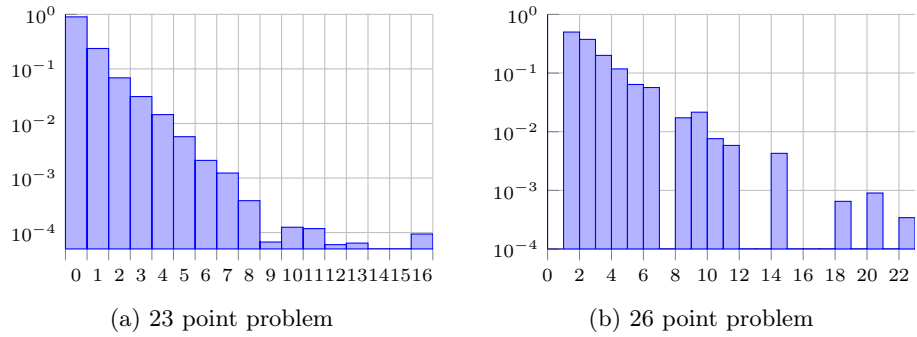
\begin{figure}[ht]
    \centering
    \begin{subfigure}[b]{0.48\linewidth}
        \centering
        \begin{tikzpicture}
\begin{axis}[
    width=1.1\linewidth,
    height=0.8\linewidth,
    ymode=log,                
    log origin=infty,
    log basis y=10,           
    ybar interval, 
    ymin=5e-5,
    ymax=1,
    xmin=0,
    xmax=17,
    xtick={0,1,2,3,4,5,6,7,8,9,10,11,12,13,14,15,16,17},
    ytickten={-4, -3, -2, -1, 0},
    tick label style={font=\scriptsize},
    grid=major,
    axis lines*=none,
    fill=blue!40,
    draw=blue!60!black
]
\addplot coordinates {
    (0, 0.900000) (1, 0.236915) (2, 0.068548) (3, 0.031124) (4, 0.014530) (5, 0.005718)
    (6, 0.002103) (7, 0.001232) (8, 0.000384) (9, 0.000067) (10, 0.000125) (11, 0.000118)
    (12, 0.000060) (13, 0.000064) (14, 1e-10) (15, 1e-10) (16, 0.000094) (17, 1e-10)
};
\end{axis}
\end{tikzpicture}
        \caption{23 point problem}
        \label{fig:prob-23}
    \end{subfigure}
    \hfill
    \begin{subfigure}[b]{0.48\linewidth}
        \centering
        \begin{tikzpicture}
\begin{axis}[
    width=1.1\linewidth,
    height=0.8\linewidth,
    ymode=log,                
    log origin=infty,
    log basis y=10,           
    ybar interval, 
    x tick label as interval=false,
    ymin=1e-4,
    ymax=1,
    xmin=0,
    xmax=23,
    xtick={0,2,4,6,8,10,12,14,16,18,20,22},
    ytickten={-4, -3, -2, -1, 0},
    tick label style={font=\scriptsize},
    grid=major,
    axis lines*=none,
    fill=blue!40,
    draw=blue!60!black
]
\addplot coordinates {
    (1, 0.500000) (2, 0.375000) (3, 0.200000) (4, 0.117647) (5, 0.063830) (6, 0.056604) 
(7, 1e-10) (8, 0.017241) (9, 0.021505) (10, 0.007605) (11, 0.005831) (12, 1e-10) 
(13, 1e-10) (14, 0.004274) (15, 1e-10) (16, 1e-10) (17, 1e-10) (18, 0.000649) 
(19, 1e-10) (20, 0.000900) (21, 1e-10) (22, 0.000342) (23, 1e-10)
};
\end{axis}
\end{tikzpicture}
        \caption{26 point problem}
        \label{fig:prob-26}
    \end{subfigure}

    \caption{Probability of being a solution versus number of violations}
    \label{fig:sols-violations-prob}
\end{figure}

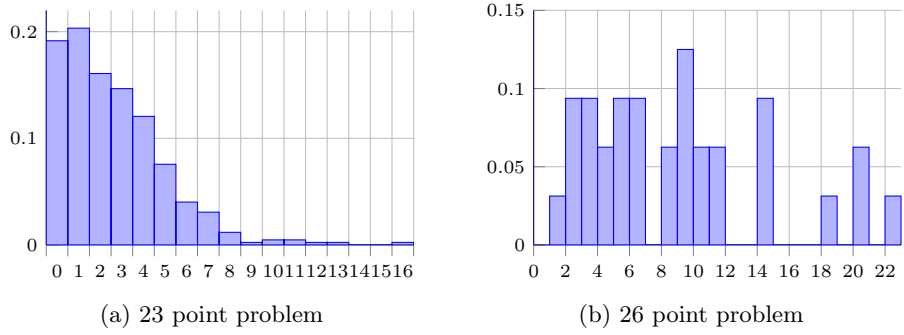
\begin{figure}[ht]
    \centering
    \begin{subfigure}[b]{0.48\linewidth}
        \centering
        \begin{tikzpicture}
\begin{axis}[
    width=1.1\linewidth,
    height=0.8\linewidth,
    ybar interval, 
    ymin=0,
    ymax=0.22,
    ytick={0,0.1,0.2},
    xmin=0,
    xmax=17,
    xtick={0,1,2,3,4,5,6,7,8,9,10,11,12,13,14,15,16,17},
    tick label style={font=\scriptsize},
    grid=major,
    axis lines*=none,
    fill=blue!40,
    draw=blue!60!black
]
\addplot coordinates {
    (0, 0.191489) (1, 0.203310) (2, 0.160757) (3, 0.146572) (4, 0.120567) (5, 0.075650)
    (6, 0.040189) (7, 0.030733) (8, 0.011820) (9, 0.002364) (10, 0.004728) (11, 0.004728)
    (12, 0.002364) (13, 0.002364) (14, 0.000000) (15, 0.000000) (16, 0.002364) (17, 0.000000)
};
\end{axis}
\end{tikzpicture}
        \caption{23 point problem}
        \label{fig:proportion-23}
    \end{subfigure}
    \hfill
    \begin{subfigure}[b]{0.48\linewidth}
        \centering
        \begin{tikzpicture}
\begin{axis}[
    width=1.1\linewidth,
    height=0.8\linewidth,
    ybar interval, 
    x tick label as interval=false,
    ymin=0,
    ymax=0.15,
    scaled y ticks=false,
    yticklabel style={/pgf/number format/fixed, /pgf/number format/precision=3},
    xmin=0,
    xmax=23,
    xtick={0,2,4,6,8,10,12,14,16,18,20,22},
    tick label style={font=\scriptsize},
    grid=major,
    axis lines*=none,
    fill=blue!40,
    draw=blue!60!black
]
\addplot coordinates {
    (1, 0.031250) (2, 0.093750) (3, 0.093750) (4, 0.062500) (5, 0.093750) (6, 0.093750) 
(7, 0.000000) (8, 0.062500) (9, 0.125000) (10, 0.062500) (11, 0.062500) (12, 0.000000) 
(13, 0.000000) (14, 0.093750) (15, 0.000000) (16, 0.000000) (17, 0.000000) (18, 0.031250) 
(19, 0.000000) (20, 0.062500) (21, 0.000000) (22, 0.031250) (23, 0.000000)
};
\end{axis}
\end{tikzpicture}
        \caption{26 point problem}
        \label{fig:proportion-26}
    \end{subfigure}

    \caption{Proportion of solutions with a given number of violations}
    \label{fig:sols-violations-prop}
\end{figure}

\section{Related Work}

Some previous work has applied SMT-style solvers to make progress on open mathematical problems. One paradigm is \emph{SAT+CAS}~\cite{satcas1,satcas2}, which combines SAT solvers with a computer algebra system (CAS) to solve problems involving symbolic computation. In this approach, partial candidate solutions are passed from the SAT solver to the CAS; if the CAS detects a problem with the candidate, a theory lemma is added to the SAT instance. This approach has been applied to study small cases of the Williamson conjecture~\cite{williamson}, a problem from design theory, and to improve the bounds on the minimum Kochen--Specker problem~\cite{kochen-specker}, a problem from quantum mechanics. Another SMT-style paradigm is \emph{satisfiability modulo symmetries} (SMS), in which a SAT solver passes partial candidate solutions to a program that checks if the candidate is in canonical form; if not, a theory lemma is generated. This approach has been applied to prove even tighter bounds on the minimum Kochen--Specker problem~\cite{sms-kochen} and verify small cases of the Murty--Simon conjecture in graph theory~\cite{sms}.

\section{Conclusion and future work} \label{sec-conclusion}
We introduced \textsf{PointSAT}, a program for finding constructions in discrete geometry by interfacing between a SAT solver and \textsf{Localizer}. We demonstrated the effectiveness of \textsf{PointSAT} by resolving an open problem about the existence of 23 points with no 6-hole or 7-gon; in fact, we can find thousands of witnessing point sets, while a previous approach was unable to find even one. Additionally, we showed that \textsf{PointSAT} is applicable to other problems in discrete geometry, which we hope will make it a useful tool for researchers in the area. The main conceptual insight, described in \Cref{sec-omit-flip}, is that solutions to these types of problems have many flippable orientations, and that ignoring these orientations when searching for a solution makes it much easier to find one. This insight has the potential to be useful for a variety of problems in discrete geometry, and its discovery is yet another success story for the application of SAT to experimental mathematics.

Despite the promising initial progress, there is still room for further improvements. For instance, \textsf{PointSAT} is still not competitive with Overmars' dedicated program for finding 29 points with no 6-hole, and it does not scale well to problems with more than 30 points. The biggest limitation of our approach seems to be that the SAT solver is still groping in the dark so to speak, finding abstract solutions without regard for how likely they are to be realizable. Figuring out how to overcome this limitation is a promising future direction.

More broadly, it may be fruitful to find other domains in which a methodology similar to the one described here could be applied. Specifically, we expect that similar approaches could be applicable to classes of problems in which there is an underapproximation that naturally encodes into SAT and a tool for testing whether a solution to the underapproximation is a true solution. Pursuing these lines of research has the potential to make new classes of mathematical problems amenable to computer-assisted discovery.

\begin{credits}
\subsubsection{\ackname} We thank Ruben Martins and Bernardo Subercaseaux for helpful discussions.
The research was supported by the National Science Foundation under grant DMS-2434625. Krapivin was supported by the NSF grant CNS2504471, Packard Foundation grant 2020-71730, and the Jeanne B. and Richard F. Berdik ARCS Pittsburgh Endowed Scholar Award. Przybocki was supported by the NSF Graduate Research Fellowship Program under Grant No. DGE-2140739.

\subsubsection{\discintname}
The authors have no competing interests to declare that are
relevant to the content of this article.
\end{credits}

\bibliographystyle{splncs04}
\bibliography{bib}
\end{document}